\documentclass[a4paper,12pt]{amsart}
\usepackage[numbers, sort&compress]{natbib}
\usepackage[left=0.9in, right=0.9in, , bottom = 1.2in, top = 1.2in]{geometry}
\usepackage[utf8]{inputenc}
\usepackage[all]{xy}
\usepackage{amssymb}
\usepackage{amsmath}
\usepackage{csquotes}
\usepackage{enumerate}
\usepackage{graphicx}
\usepackage{amsaddr}
\usepackage{color}
\usepackage{mwe}
\usepackage{subfig}
\usepackage{appendix}
\usepackage{multirow}

\usepackage{soul}
\usepackage{xcolor}
\newcommand\independent{\protect\mathpalette{\protect\independenT}{\perp}}
\def\independenT#1#2{\mathrel{\rlap{$#1#2$}\mkern2mu{#1#2}}}

\usepackage{tikz}
\usetikzlibrary{arrows.meta,shapes}
\usetikzlibrary{arrows,shapes.arrows,shapes.geometric,shapes.multipart,
decorations.pathmorphing,positioning,shapes.swigs,}
\usepackage{setspace}
\newtheorem{lemma}{Lemma}

\newtheorem{proposition}{Proposition}
\newtheorem{definition}{Definition}
\newtheorem{assumption}{Assumption}

\MakeOuterQuote{"}\EnableQuotes
\DeclareUnicodeCharacter{00A0}{ }

\title[]{Effects conditional on post-treatment events generated by independent mechanisms}

\author{Marco Piccininni$^{1}$, Mats J. Stensrud$^{1}$} \address{ \small $^1$ Institute of Mathematics, École Polytechnique Fédérale de Lausanne, Switzerland \\}

\begin{document}

\begin{abstract} In both observational studies and randomized trials, post-treatment events such as dropout, nonadherence, and truncation by death occur frequently. In some studies, conditioning on post-treatment variables is a deliberate strategy to isolate particular treatment effects on the outcome. However, naive comparisons of outcomes conditional on post-treatment events generally lack a causal interpretation, even when treatment is randomly assigned. There exist causal estimands that account for post-treatment events, including survivor average causal effects and conditional separable effects, but identification usually requires measurement of common causes of the post-treatment event and the outcome. In this article, we show that such measurements are not always necessary. Conceptually, what we require is that the treatment and other unmeasured causes of the outcome generate the post-treatment event through ``independent mechanisms''. Then, conditional separable effects and survivor average causal effects can be identified without adjustment for common causes of the post-treatment event and the outcome. We illustrate the results in studies with truncating events, differential nonadherence, and the birth weight paradox.
\end{abstract}

\maketitle


\clearpage

\section{Introduction}
\label{sec: intro}

Many causal analyses require explicit consideration of post-treatment events. Even in randomized controlled trials (RCTs), researchers need to account for individuals who drop out of the study or individuals who experience competing events. Sometimes conditioning on post-treatment variables is also a deliberate choice, for example, when attempting to isolate the direct effect of the treatment on the outcome.

However, conditioning on a post-treatment variable leads to problems, even in randomized trials. To fix ideas, suppose the post-treatment variable is affected by treatment, and shares common causes with the outcome. Then, we expect a non-null association between the treatment and the outcome conditional on the post-treatment variable even if the treatment does not cause the outcome for any unit\cite{hernan2010causal}. This phenomenon is often referred to as ``collider bias''\cite{pearl2018book, holmberg2022collider,banack2024collider}, referring to the graphical relationships of such a post-treatment variable in a causal directed acyclic graph (DAG)\cite{banack2024collider, holmberg2022collider}. The birth weight paradox is a canonical example of a counter-intuitive finding attributed to such collider bias\cite{hernandez2006birth, pearl2018book, vanderweele2014commentary, banack2024collider}.

Several causal estimands have been proposed to target effects that account for post-treatment events\cite{robins1992identifiability}, such as separable effects\cite{robins2010alternative,robins2022interventionist,stensrud2022separable, stensrud2021generalized,stensrud2023conditional,hernan2010causal}. Conditional separable effects can allow the outcome to be undefined if the post-treatment event happens\cite{stensrud2023conditional,young2021identified}, which is, e.g., relevant in truncation by death settings. The existing identification results for these effects require measurements of confounders of the post-treatment variable and the outcome\cite{stensrud2023conditional,malinsky2019potential} or proxy variables\cite{park2024proximal}.

In this work, we present new results ensuring that, in certain settings, conditional separable effects are identified even without measuring confounders of the post-treatment and outcome variable. This is possible when the mechanism through which treatment affects the post-treatment event is independent of mechanisms involving other causes of the outcome. The result builds on a known property of survival multiplicative models\cite{hernan2010causal, hernan2004structural}. We present several illustrative examples: studies where participants experience truncating events, studies with differential non-adherence, and the famous birth weight paradox. Furthermore, we show that a similar result holds for the survivor average causal effect\cite{robins1986new, rubin2006causal, stensrud2023conditional}, and we also draw connections to traditional instrumental variable analysis\cite{angrist1995identification}.

\section{Setup and notation}

Consider an experiment in which individuals are randomly assigned to treatment ($A=1$) or control ($A=0$). The outcome of interest is a numerical variable $Y$, and let $D \in \{0,1\}$ be an indicator of a post-treatment event, possibly affected by $A$ and affecting $Y$. Let $U$ denote exogenous, unmeasured common causes of $D$ and $Y$, which might exist because only $A$ is randomly assigned in our setting. To simplify the presentation, we assume $U$ to be discrete.

The causal model is represented in the DAG in Figure \ref{fig:dag_general}, and we will assume that the graph encodes a Finest Fully Randomized Causally Interpreted Structured Tree Graph model (FFRCISTG)\cite{robins1986new, richardson2013single}.

\begin{figure}[!ht]
\centering
    \begin{tikzpicture}
       \begin{scope}[every node/.style={thick,draw}]
       \node[name=A, shape=ellipse] at (-2,0){$A$};
       \node[name=D, shape=ellipse] at (0,0){$D$};
       \node[name=Y, shape=ellipse] at (3,0){$Y$};
       \node[name=U, shape=ellipse, dashed] at (1.5,2){$U$};
    \end{scope}
        
    \begin{scope}[>={Stealth[black]},
                  every node/.style={fill=white,circle},
                  every edge/.style={draw=black,very thick}]
        \path[->] (U) edge (Y);
        \path[->] (U) edge (D);            
        \path[->] (A) edge (D);
        \path[->] (D) edge (Y);
        \path[->] (A) edge[bend right] (Y);
    \end{scope}
    \end{tikzpicture}
\caption{Causal directed acyclic graph representing the assumed relationship between treatment ($A$), outcome ($Y$), post-treatment variable ($D$) and potential common causes of $D$ and $Y$ ($U$). Dashed nodes represent unmeasured variables.}
\label{fig:dag_general}
\end{figure}
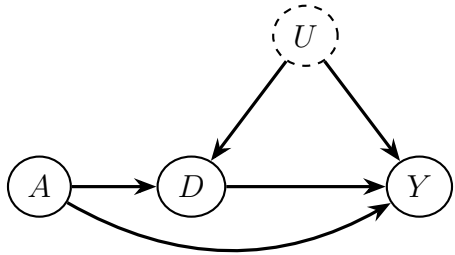

In such a trial the average (marginal) effect of $A$ on $Y$ is easy to identify when $Y$ is measured in all units; because there is no confounding the association between the treatment and the outcome represents the marginal effect\cite{hernan2010causal}.

Suppose that we are interested in isolating the effect of $A$ on $Y$, for example, because the outcome $Y$ is undefined, or not of interest, when the post-treatment event $D$ occurs. More formally, we represent this by giving the variable $Y$ an arbitrary value when $D=1$, indicating an ``undefined'' state. This representation aligns with the causal DAG in Figure \ref{fig:dag_general}, in which the edge $D \rightarrow Y$ allows the distribution of $Y$ to depend on $D$ given $A$ and $U$. Marginal expectations of $Y$ are therefore not of practical interest in this setting, and several different causal estimands have been suggested to quantify effects, see for example \cite{robins1986new, rubin2006causal,stensrud2023conditional,wang2017identification}. The existing strategies require access to measurements of the common causes $U$ or their proxies\cite{park2024proximal}. However, we will show that when $D$ is generated in a particular way that can sometimes be justified in practice, certain causal effects can be estimated without adjusting for $U$.

To formalize the causal parameters and clarify our assumptions we will rely on potential outcomes\cite{hernan2010causal}. In particular, we will use notation $Y^{a=a'}$ to indicate the value of the outcome had, possibly contrary to fact, the treatment variable $A$ been set to $a'$.

While we use the story of a randomized trial to simplify the exposition, our reasoning is easy to  adapt to an observational study where our causal model is plausible after having restricted the analysis on some values of baseline covariates.

\section{Independent mechanisms}

Consider a large RCT that compares quality of life at 6 months ($Y$) between individuals who undergo a plastic surgery ($A=1$) and no treatment ($A=0$). The surgery takes place immediately after randomization. Some participants died before 6 months of follow-up ($D=1$). Thus, the outcome $Y$ is only defined for individuals who survived until the end of the study ($D=0$). The post-treatment event $D$ here represents what is called a ``truncating'' event\cite{stensrud2023conditional}, sometimes broadly classified as an ``intercurrent event''\cite{stensrud2022translating, ich2019addendum} or a ``competing event''\cite{cottone2023estimand}, although the precise definition of these terms is debated and varies across literatures\cite{young2021identified}.  The plastic surgery of interest is only intended to have aesthetic effects, but the procedure has some risks and leads to a higher mortality rate in the treatment arm compared to the control arm. We also expect unmeasured common causes ($U$) of death and quality of life (e.g., comorbidities, socioeconomic status etc.).

The data generating process of this hypothetical trial is reflected in the causal DAG in Figure \ref{fig:dag_general}. Without further assumptions, the comparison of outcomes under treatment and control, restricted to survivors, cannot be interpreted causally: because $D$ is affected by treatment and shares common causes with the outcome, conditioning on $D$ generally induces a non-causal association between treatment and outcome, that is, collider stratification bias\cite{hernan2010causal, hernan2004structural, banack2024collider}.

Suppose that, during the study period, treatment only affects mortality through adverse events, for example because the anesthetic can induce fatal reactions. Let $D_A \in \{0,1\}$ indicate a fatal anesthetic-related adverse event that occurred in a short time window after randomization. For our methodological results, it is important that $D_A$ mediates all the effect of $A$ on $D$. In particular, the plastic surgery only affects mortality through the anesthetic-related adverse events. Although $D_A$ is unmeasured in the study, it will have an important role in our arguments.

We assume that $D_A$ is independent of other unmeasured causes of $D$ or $Y$, as described in Figure \ref{fig:dag_augmented}, which is an augmented\cite{hernan2010causal} version of the DAG in Figure \ref{fig:dag_general}. In Figure \ref{fig:swig_augmented} we present the corresponding single world intervention graph (SWIG) representing the intervention $A=a'$ \cite{richardson2013single}.

In the surgery example, this causal structure is plausible if fatal anesthetic-related adverse events are unrelated to characteristics of the participants that affect other causes of death or future quality of life. For example, certain severe reactions to anesthesia are thought to be induced by genetic mutations and this susceptibility does not affect aspects of everyday life\cite{schneiderbanger2014management}.

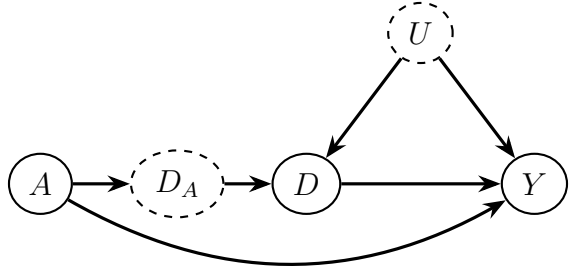
\begin{figure}[!ht] 
\centering
    \begin{tikzpicture}
       \begin{scope}[every node/.style={thick,draw}]
       \node[name=A, shape=ellipse] at (-3.5,0){$A$};
       \node[name=D, shape=ellipse] at (0,0){$D$};
       \node[name=D_A, shape=ellipse, dashed] at (-1.7,0){$D_A$};
       \node[name=Y, shape=ellipse] at (3,0){$Y$};
       \node[name=U, shape=ellipse, dashed] at (1.5,2){$U$};
    \end{scope}
        
    \begin{scope}[>={Stealth[black]},
                  every node/.style={fill=white,circle},
                  every edge/.style={draw=black,very thick}]
        \path[->] (U) edge (Y);
        \path[->] (U) edge (D);            
        \path[->] (A) edge (D_A);
        \path[->] (D_A) edge (D);
        \path[->] (D_A) edge[bend right=20] (Y);
        \path[->] (D) edge (Y);
        \path[->] (A) edge[bend right=35] (Y);
    \end{scope}
    \end{tikzpicture}
    \caption{Causal directed acyclic graph representing the assumed relationship between treatment ($A$), outcome ($Y$), post-treatment variable ($D$) and potential common causes of $D$ and $Y$ ($U$). The node $D_A$ represents the specific type of post-treatment event that can be affected by the treatment. Dashed nodes represent unmeasured variables.}
    \label{fig:dag_augmented}
\end{figure}

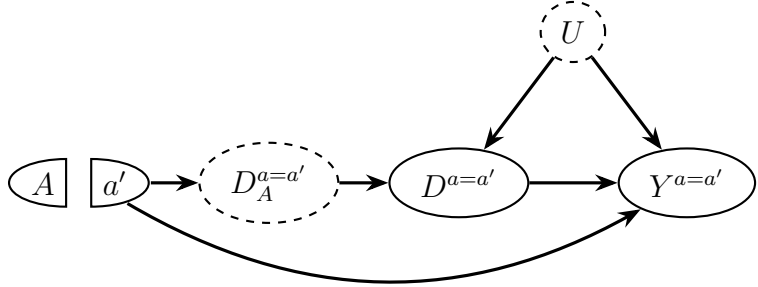
\begin{figure}[!ht] 
\centering
    \begin{tikzpicture}
       \begin{scope}[every node/.style={thick,draw}]
       \node[name=A,shape=swig vsplit] at (-5,0) {
                                              \nodepart{left}{$A$}
                                              \nodepart{right}{$a'$} };
       \node[name=D, shape=ellipse] at (0,0){$D^{a=a'}$};
       \node[name=D_A, shape=ellipse, dashed] at (-2.5,0){$D^{a=a'}_A$};
       \node[name=Y, shape=ellipse] at (3,0){$Y^{a=a'}$};
       \node[name=U, shape=ellipse, dashed] at (1.5,2){$U$};
    \end{scope}
        
    \begin{scope}[>={Stealth[black]},
                  every node/.style={fill=white,circle},
                  every edge/.style={draw=black,very thick}]
        \path[->] (U) edge (Y);
        \path[->] (U) edge (D);            
        \path[->] (A) edge (D_A);
        \path[->] (D_A) edge (D);
        \path[->] (D_A) edge[bend right=20] (Y);
        \path[->] (D) edge (Y);
        \path[->] (A) edge[bend right] (Y);
    \end{scope}
    \end{tikzpicture}
    \caption{Single world intervention graph corresponding to the directed acyclic graph in Figure \ref{fig:dag_augmented} for the intervention $A=a'$.}
    \label{fig:swig_augmented}
\end{figure}

Suppose that the causal graphs in Figure \ref{fig:dag_augmented} and \ref{fig:swig_augmented} represent our setting of interest. More formally, consider the following assumption:
\begin{assumption} \label{ass:swig_augmented}
We assume a FFRCISTG model associated with the graph in Figure \ref{fig:dag_augmented}.
\end{assumption}
It follows from Assumption \ref{ass:swig_augmented}  that the distribution of counterfactual variables factorizes according to the SWIG in Figure \ref{fig:swig_augmented} and modularity holds\cite{richardson2013single}.

We also invoke a mild positivity\cite{hernan2010causal} condition, requiring that in each treatment arm some participant does not experience the post-treatment event:
\begin{assumption} \label{ass:positivity}
$\Pr(A=a)>0$ and $\Pr(D=0 \mid A=a)>0$ for any $a \in \{0,1\}.$
\end{assumption}

Because $D_A$ indicates the occurrence of a subtype of the post-treatment event $D$, we have the following determinism:
\begin{assumption} \label{ass:D_Asufficient_po}
If $D_A^{a=a'}=1$ then $D^{a=a'}=1$, for every $a'$.
\end{assumption}
Assumption \ref{ass:D_Asufficient_po} simply states that when a participant dies due to anesthetic-related adverse events, the participant dies. This is an uncontroversial assumption in our surgery example.

\begin{lemma} \label{lemma:multiplicative_survival}
Under Assumptions \ref{ass:swig_augmented}, \ref{ass:positivity}, and \ref{ass:D_Asufficient_po}, we have that for every $a' \in \{0,1\}$
$$\Pr(D^{a=a'}=0 \mid U=u)= \Pr(D=0 \mid U=u, D_A=0) \Pr(D_A=0 \mid A=a').$$
\end{lemma}
See Appendix \ref{app:proofs} for a formal proof.

Lemma \ref{lemma:multiplicative_survival} gives conditions under which $D$ follows a so-called ``multiplicative survival model''\cite{hernan2004structural,hernan2010causal}: the probability of ``surviving'' ($D=0$) when $U=u$ and $A$ is fixed to $a'$ is equal to the multiplication of a function of $a'$ and a function of $u$. Conceptually, the multiplicative survival model describes survival as the result of avoiding death from two independent mechanisms: one involving only $A$, and one involving only $U$. We can interpret Assumptions \ref{ass:swig_augmented} and \ref{ass:D_Asufficient_po}, together, as a claim about an independence of the mechanism through which $A$ causes $D$ and the mechanism through which $U$ causes $D$. To fix ideas, we give an example in which a multiplicative survival model is induced by a sufficient-component cause model\cite{hernan2010causal} where there is no ``interaction'' between $A$ and $U$ in determining $D$ (see Appendix \ref{app:causal_pies}).

In multiplicative survival models independent causes of death remain independent when conditioning on survival, despite the collider structure \cite{hernan2004structural, hernan2010causal}. Specifically, the distribution of $U$ does not differ under different interventions on $A$, even when conditioning on $D=0$.

\begin{lemma} \label{lemma:useful_factorization}
Under Assumptions \ref{ass:swig_augmented}, \ref{ass:positivity}, and \ref{ass:D_Asufficient_po}, we have that
\begin{align*}
\Pr(U=u \mid D^{a=0}=0)=\Pr(U=u \mid D^{a=1}=0).
\end{align*} 
\end{lemma}

Lemma \ref{lemma:useful_factorization} is proven in Appendix \ref{app:proofs}.

\section{Causal effects}

In our surgery example, quality of life can only be assessed in those who remain alive. An intervention that prevents individuals from dying is practically infeasible, and causal effects requiring such interventions are not of much practical relevance\cite{stensrud2023conditional}. For example, the commonly targeted controlled direct effect (i.e., $\mathbb{E}(Y^{a=1,d=0})-\mathbb{E}(Y^{a=0,d=0})$)\cite{robins1992identifiability, hernan2010causal} representing the effect of surgery under an intervention that forces $D=0$ seems unsuitable to inform decisions when $D$ is death.

Suppose, however, that the doctors are interested in whether the surgical procedure can be improved. That is, they want to create a procedure, without dangerous complications, that otherwise gives the same aesthetic result. To formalize this line of thinking, we interpret the existing treatment $A$ as having two components\cite{robins2010alternative,stensrud2023conditional}: a component $A_Y$ that only affects $Y$ directly, and a component $A_D$ that affects $D_A$. In our example, $A_D$ could represent the chemical compound of the standard anesthetic which can induce severe adverse reactions. Instead, $A_Y$ represents the surgical operation itself, which affects aesthetic appearance, and thus quality of life. In the observed RCT, we can recode those randomized to treatment $A=1$ as receiving both $A_D=1$ and $A_Y=1$, and, similarly, those randomized to control $A=0$ as receiving both $A_D=0$ and $A_Y=0$. We assume that an intervention setting $A$ to a specific value is equivalent to setting both $A_D$ and $A_Y$ to the same value\cite{robins2010alternative,robins2022interventionist,stensrud2023conditional}. More formally, the notation $Y^{a_D=a',a_Y=a}$ represents potential outcomes for $Y$ under an intervention on both $A_D$ and $A_Y$:

\begin{assumption} \label{ass:treatment_decomposition}
Treatment $A$ can be decomposed into $A_D \in \{0,1\}$ and  $A_Y \in \{0,1\}$, with $A = A_D = A_Y$ in the observed data.
For every $a' \in \{0,1\}$ and every variable $V$ in our causal model, $V^{a_D=a',a_Y=a'}=V^{a=a'}$. 
\end{assumption}

Although $A_D$ and $A_Y$ always coincide with the value of $A$ in the data, we can envision potential interventions setting discordant values for these two variables. Consider a hypothetical four-arm trial where we randomize both $A_D$ and $A_Y$. Some participants are randomized to have the usual anesthetic and the surgery ($A_D=1,A_Y=1$), some to have neither ($A_D=0,A_Y=0$), some to undergo the surgery but with a special anesthetic drug that does not have the dangerous component ($A_D=0,A_Y=1$), and some to take the standard anesthetic without undergoing the surgery ($A_D=1,A_Y=0$). Obviously, the last two arms were not observed, but we can nevertheless envision such a trial, and we can represent the trial in a DAG and a SWIG (see Figures \ref{fig:dag_augmented_4arm} and \ref{fig:swig_augmented_4arm}, respectively).

\begin{assumption} \label{ass:swig_augmented_4arm}
For the hypothetical four-arm trial where both $A_D$ and $A_Y$ are randomized independently, assume a FFRCISTG model associated with the graph in Figure \ref{fig:dag_augmented_4arm}.
\end{assumption}

\begin{figure}[!ht] 
\centering
    \begin{tikzpicture}
       \begin{scope}[every node/.style={thick,draw}]
       \node[name=A_D, shape=ellipse] at (-3.4,0){$A_D$};
       \node[name=A_Y, shape=ellipse] at (-3.4,-1.5){$A_Y$};
       \node[name=D, shape=ellipse] at (0,0){$D$};
       \node[name=D_A, shape=ellipse, dashed] at (-1.6,0){$D_A$};
       \node[name=Y, shape=ellipse] at (3,0){$Y$};
       \node[name=U, shape=ellipse, dashed] at (1.5,2){$U$};
    \end{scope}
        
    \begin{scope}[>={Stealth[black]},
                  every node/.style={fill=white,circle},
                  every edge/.style={draw=black,very thick}]
        \path[->] (U) edge (Y);
        \path[->] (U) edge (D);
        \path[->] (A_D) edge (D_A);
        \path[->] (D_A) edge (D);
        \path[->] (D_A) edge[bend right=25] (Y);
        \path[->] (D) edge (Y);
        \path[->] (A_Y) edge[bend right] (Y);
    \end{scope}
    \end{tikzpicture}
    \caption{Causal directed acyclic graph representing a hypothetical trial in which $A_D$ and $A_Y$ are randomized independently. The nodes $A_D$ and $A_Y$ represent the components of the treatment that affect directly only $D_A$ and $Y$, respectively. Dashed nodes represent unmeasured variables.}
    \label{fig:dag_augmented_4arm}
\end{figure}
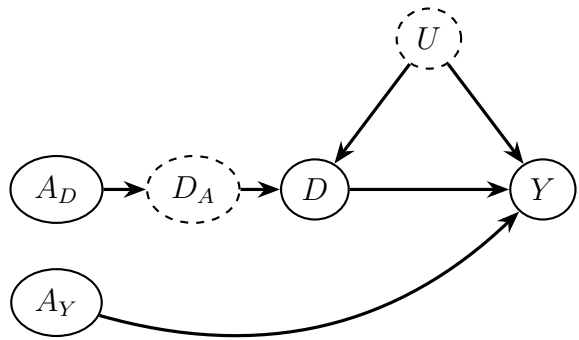

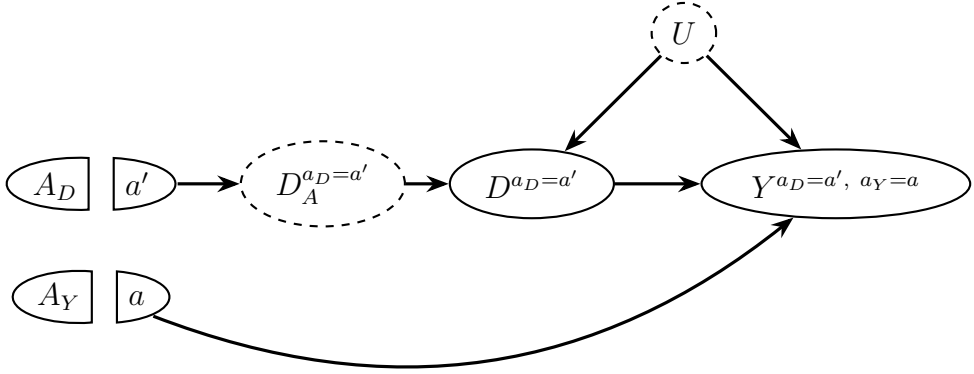
\begin{figure}[!ht] 
\centering
    \begin{tikzpicture}
       \begin{scope}[every node/.style={thick,draw}]
       \node[name=A_D,shape=swig vsplit] at (-5.8,0) {
                                              \nodepart{left}{$A_D$}
                                              \nodepart{right}{$a'$} };
        \node[name=A_Y,shape=swig vsplit] at (-5.8,-1.5) {
                                              \nodepart{left}{$A_Y$}
                                              \nodepart{right}{$a$} };
       \node[name=D, shape=ellipse] at (0,0){$D^{a_D=a'}$};
       \node[name=D_A, shape=ellipse, dashed] at (-2.75,0){$D^{a_D=a'}_A$};
       \node[name=Y, shape=ellipse] at (4,0){$Y^{a_D=a',\ a_Y=a}$};
       \node[name=U, shape=ellipse, dashed] at (2,2){$U$};
    \end{scope}
        
    \begin{scope}[>={Stealth[black]},
                  every node/.style={fill=white,circle},
                  every edge/.style={draw=black,very thick}]
        \path[->] (U) edge (Y);
        \path[->] (U) edge (D);
        \path[->] (A_D) edge (D_A);
        \path[->] (D_A) edge (D);
        \path[->] (D_A) edge[bend right=20] (Y);
        \path[->] (D) edge (Y);
        \path[->] (A_Y) edge[bend right] (Y);
    \end{scope}
    \end{tikzpicture}
    \caption{Single world intervention graph  corresponding to the directed acyclic graph in Figure \ref{fig:dag_augmented_4arm} representing the intervention on the components of treatment $A_D$ and $A_Y$ in the hypothetical four-arm trial. Dashed nodes represent unmeasured variables.}
    \label{fig:swig_augmented_4arm}
\end{figure}

Because the outcome is only defined when $D=0$, a sensible causal effect is the conditional separable effect.

\begin{definition} \label{def:CSE}
    $$CSE(a')=\mathbb{E}(Y^{a_D=a',a_Y=1} \mid D^{a_D=a'}=0 )-\mathbb{E}(Y^{a_D=a',a_Y=0} \mid D^{a_D=a'}=0 ).$$
\end{definition}

The $CSE(a')$ represents the effect of manipulating the component $A_Y$  when keeping the component $A_D$ fixed to the value $a'$, only among individuals who do not experience the post-treatment event under such interventions. In our example, doctors are interested in the effect of plastic surgery on quality of life using a new anesthetic that does not have the dangerous component ($A_D=0,A_Y=1$) against receiving no treatment at all ($A_D=0,A_Y=0)$, among survivors. This contrast corresponds to the $CSE(0)$. The $CSE(0)$ can be estimated in a future experiment where the new anesthetic drug without side effects is actually used\cite{robins2010alternative,stensrud2023conditional}. 

\subsection{Vaccine trial example}
Our assumptions are not always reasonable and need to be justified by subject-matter arguments. Consider, for example, a double-blind RCT, similar to \cite{sadoff2021interim}, that compares the level of antibodies against SARS-CoV-2 at 15 days ($Y$) in seronegative individuals receiving an adenoviral vector‐based COVID-19 vaccine ($A=1$) or placebo ($A=0$) at baseline. Adenoviral vector‐based COVID-19 vaccines can cause fatal vaccine–induced immune thrombotic thrombocytopenia ($D_A=1$)\cite{greinacher2022vaccine, cines2021sars}. Despite being rare, through this side effect, the treatment can increase the probability of death during the study period ($D=1$). Assumptions \ref{ass:positivity} and \ref{ass:D_Asufficient_po} are satisfied by design, and the decomposition in Assumption \ref{ass:treatment_decomposition} is plausible: we just need to consider a modification of the adenoviral-vector platform to remove or neutralize the specific feature responsible for the fatal side effect ($A_D$), while preserving the vaccine component responsible for inducing SARS-CoV-2 antibodies ($A_Y$). 

Now, the $CSE(0)$ corresponds to the benefit of giving a new adenoviral vector-based vaccine, without the dangerous component versus placebo, to those who survive. This is an effect of practical interest, as researchers are investigating the possibility of developing such vaccines\cite{sallard2026novel}.

Our DAG in Figure \ref{fig:dag_augmented_4arm} requires that the fatal side effect is not associated with other causes of SARS-CoV-2 antibodies or other causes of death. This seems to be plausible in a population with similar characteristics (e.g., similar age), because the vaccine–induced immune thrombotic thrombocytopenia is presumed to be caused by the combination of a genetic predisposition and a specific somatic hypermutation\cite{wang2026adenoviral}. Also, according to the DAG, $A_Y$ has no effect on $D$. This is plausible in our 15-day long study, but would not be plausible in a longer study. While the effect of the vaccine active component on mortality is negligible in the first two weeks, it will not be after 15 days due to its protective effect against COVID-19 disease.

\section{Identification}
Unlike existing identification results on separable effects\cite{robins2010alternative,robins2022interventionist,stensrud2023conditional,park2024proximal}, we can identify the conditional separable effect of $A_Y$ on the outcome without measuring $D_A$ or the unmeasured confounders of the $D\text{-}Y$ relationship ($U$).

\begin{proposition} \label{prop:cse_identification}
Under Assumptions \ref{ass:positivity}, \ref{ass:D_Asufficient_po}, \ref{ass:treatment_decomposition}, and \ref{ass:swig_augmented_4arm} we have
$$CSE(0)=CSE(1)=\mathbb{E}(Y \mid A=1, D=0)-\mathbb{E}(Y \mid A=0, D=0)$$
\end{proposition}
A formal proof is presented in Appendix \ref{app:proofs}.
Proposition \ref{prop:cse_identification} states that both conditional separable effects correspond to the association between the treatment and the outcome among the individuals who have $D=0$ in the observed data.

This result might seem surprising, as the DAG in Figure \ref{fig:dag_general}, from which our reasoning started, suggests that conditioning on the collider $D$ induces a non-causal association between the treatment and the outcome. However, our story-led assumptions induced a multiplicative survival model (Lemma \ref{lemma:multiplicative_survival}) which has interesting properties when conditioning on survival (see Lemma \ref{lemma:useful_factorization}). In Proposition \ref{prop:cse_identification}, we exploited these properties to identify the conditional separable effects without requiring additional measurements or further assumptions about the unmeasured confounders. To develop an intuition of our identification strategy, consider the DAG in Figure \ref{fig:dag_augmented_4arm}. When we condition on $D=0$, we are also implicitly conditioning on $D_A=0$, because of the deterministic relation between them, and we block all paths between $A_D$ and $Y$. Thus, when conditioning on $D=0$, the association between the treatment and the outcome in the data represents the direct effect of $A_Y$ on $Y$.

Our main interest is conditional separable effects, as we believe they are more attractive, interventionist causal effects in our setting. However, we recognize that the ICH E9(R1) addendum to the guideline on statistical principles for clinical trials\cite{ich2019addendum} has popularized different effects, including principal stratum estimands such as the survivor average causal effect ($SACE$)\cite{robins1986new, rubin2006causal, stensrud2023conditional}:

\begin{definition}
    $SACE=\mathbb{E}(Y^{a=1} \mid D^{a=0}=D^{a=1}=0)-\mathbb{E}(Y^{a=0} \mid D^{a=0}=D^{a=1}=0).$
\end{definition}

The $SACE$ is the effect of treatment among individuals who would not experience the post-treatment event during the study, regardless of which treatment they would receive. In our plastic surgery example, this corresponds to the effect of the surgery among individuals who would not die during the study period regardless of the assigned treatment. In Appendix \ref{app:sace}, we show that, under Assumptions \ref{ass:swig_augmented}, \ref{ass:positivity} and \ref{ass:D_Asufficient_po}, a monotonicity assumption, and a cross-world independence conditional on $U$, the $SACE$ is identified by the functional given in Proposition \ref{prop:cse_identification}, thereby bridging identification results on the $CSE(a')$ and the $SACE$.

Finally, in Appendix \ref{app:proof_identification_L} we give an identification result for the $CSE(a')$ in the setting where a (measured) variable $L$ causes $D_A,D,$ and $Y$, and is affected by $A_D$. This relaxes the assumption that $A_D$ only affects $Y$ through $D_A$, and that $D_A$ and $D,Y$ share no common causes.

\section{Independent mechanisms for adherence}

Consider an RCT including pregnant women who want to stop smoking. This example is inspired by the RCT conducted by Hajek et al.\cite{hajek2022electronic}. Women were randomly assigned to e-cigarettes ($A=0$) or nicotine replacement therapy (NRT, $A=1$). In the e-cigarettes arm, participants were sent e-cigarette starter kits, offered an 8-week supply, and encouraged to tailor the e-cigarette use to their needs\cite{hajek2022electronic}. In the NRT arm, participants were sent a 2-week nicotine patches supply, offered an 8-week supply, and encouraged to also use other NRT products if needed\cite{hajek2022electronic}. The participants received behavioral support through phone calls in both arms. 
Here, we consider self-reported abstinence at 4 weeks as the outcome of interest $Y$, which was one of the secondary outcomes in the original study\cite{hajek2022electronic}. The authors found a negative effect of being randomized to NRT compared to e-cigarettes on $Y$ ($10.7\%$ versus $15.6\%$\cite{hajek2022electronic}).
Let $D$  indicate whether a participant deviated from the assigned treatment ($D=1$) or adhered with the assigned treatment ($D=0$) for the first 4 weeks. Adherence was lower in the NRT arm compared to the e-cigarettes arm. At 4 weeks, $39.9\%$ of participants in the e-cigarettes arm and $22.5\%$ in the NRT arm were currently using the allocated product\cite{hajek2022electronic}. E-cigarettes and NRT have very different adverse reactions. NRT's adverse reactions were primarily skin irritation and nausea; while e-cigarettes' adverse reactions were primarily throat irritation and cough\cite{hajek2022electronic}. Importantly, an inclusion criterion for the study was to be willing to use either NRT or e-cigarettes (``no strong preference''\cite{hajek2022electronic}), and agreeing to use only the allocated stop-smoking product (and not the other) for at least the first 4 weeks\cite{hajek2022electronic}.
Given the lack of preference and the commitment, it seems plausible that the difference in adherence between the arms was only due to the occurrence of events perceived as side effects from the allocated product, or fear of side effects. We indicate with $D_A$ non-adherence in the first 4 weeks due to these reasons. In line with our argument, more participants interrupted product use due to adverse reactions attributable to the product in the NRT arm\cite{hajek2022electronic}.  We assume that $D_A$ is independent of other unmeasured causes of non-adherence or abstinence at 4 weeks. This can be plausible if, for example, strong side effects appear after first uses and their appearance depends only on some predisposition (e.g., skin sensitivity) that is not related to other personal characteristics. In this story, our Assumptions \ref{ass:swig_augmented}, \ref{ass:positivity} and \ref{ass:D_Asufficient_po} are plausible.

Further suppose that $D_A$ is actually only caused by the assigned treatment through skin irritation or fear of skin irritation. This is not implausible, as the side effects of e-cigarettes, such as throat irritation and cough are similar to side effects of regular smoking, and unlikely to affect product use in a population of smokers. Nausea is common among pregnant women in weeks 12-24 of gestation and is also induced by e-cigarettes or regular smoking; therefore, nausea is not likely to be a determining factor for adherence decision caused by treatment assignment. Suppose that the researchers are interested in the effect of an altered nicotine patch that does not irritate the skin. This reflects a question where we decompose the assigned treatment into nicotine patch delivery system ($A_Y=1$) and e-cigarette ($A_Y=0$), and skin irritant ($A_D=1$) or no/neutralized skin irritant ($A_D=0$). These arguments justify the assumptions specific to separable effects (Assumption \ref{ass:treatment_decomposition} and \ref{ass:swig_augmented_4arm}). 

According to Proposition \ref{prop:cse_identification}, the investigators can interpret the association between assigned treatment and outcome, estimated only among adherers ($D=0$) as a valid effect estimate of $CSE(0)$. This association quantifies the effect of assigning a new nicotine patch, with the same nicotine delivery system of the real patch but without skin irritant, versus e-cigarette, among adherers. We want to emphasize that this conclusion is agnostic towards the specific consequences of non-adherence: individuals who don't adhere ($D=1$) can take a different stop-smoking product than the one assigned, can refuse any help, or even drop out of the study.

In this setting, the $SACE$ represents the effect of assigned treatment among compliers ($D^{a=0}=D^{a=1}=0$), an estimand that is often targeted in instrumental variable analyses\cite{angrist1995identification}. We expand on the difference with traditional instrumental variable approaches in Appendix \ref{app:instrumental_link}.

We emphasize that the investigators of the RCT did not believe the association between treatment and outcome among adherers to have a causal interpretation\cite{hajek2022electronic}. This might have been due to their belief that non-adherence was also affected by previous success in reducing smoking\cite{hajek2022electronic}.

\section{The birth weight paradox revisited}

In 1971, Yerushalmy  published an article reporting some empirical findings about the association between having a smoking mother and infant mortality\cite{yerushalmy1971relationship}. These findings were considered  ``paradoxical'' by Yerushalmy\cite{yerushalmy1971relationship}. Specifically, women who smoke have higher risk of having low-birth-weight babies (i.e., weight $\leq 2,500$ grams), compared to women who do not smoke. Overall, the mortality of infants born from mothers who smoke is higher than the mortality of infants born from mothers who do not smoke. Yet, low-birth-weight infants of smoking mothers have lower mortality compared to low-birth-weight infants of non-smoking mothers. This finding could naively be interpreted as a desirable effect of smoking in low-birth-weight babies\cite{wilcox2001importance, hernandez2006birth}, and became  known as the ``birth weight paradox''\cite{wilcox2001importance, hernandez2006birth, pearl2018book, schisterman2009z, whitcomb2009quantification, vanderweele2014commentary}. Similar associations have been reported in other studies and also for different exposures than smoking\cite{wilcox2001importance, schisterman2009z, whitcomb2009quantification, vanderweele2014commentary}. The paradox puzzled epidemiologists for decades\cite{joseph2004parsimonious, pearl2018book, whitcomb2009quantification} and no intuitive explanation was presented until the first decade of the twenty-first century\cite{joseph2004parsimonious, pearl2018book}. In 2006, Hern{\'a}ndez-D{\'\i}az et al.\cite{hernandez2006birth} argued that the paradox could be explained by collider stratification bias. Since then, the collider bias explanation gained traction and stimulated more research on the topic\cite{schisterman2009z, whitcomb2009quantification, vanderweele2014commentary}.

Nevertheless, there are still some open questions regarding the birth weight paradox. First, Wilcox\cite{wilcox2006invited} pointed out that the paradox is only seen in low-birth-weight babies. The conditional association between smoking and infant mortality agrees with the biological understanding (smoking is harmful) among high-birth-weight babies.
Second, several authors, including Hern{\'a}ndez-D{\'\i}az et al., have argued that one purpose of conditioning on birth weight was to isolate a direct effect of smoking on infant mortality not mediated by birth weight\cite{schisterman2009z, whitcomb2009quantification, vanderweele2014commentary, hernandez2006birth}. Yet, they have cautioned against the interpretation of such analyses in practice, precisely because the presence of unmeasured common causes of birth weight and mortality can bias estimators of the direct effect\cite{schisterman2009z,vanderweele2014commentary,hernandez2006birth}.

Our work contributes to this debate. Both Hern{\'a}ndez-D{\'\i}az et al. (Figure 3.7 in \cite{hernandez2006birth}) and VanderWeele (Figure 2b-c in \cite{vanderweele2014commentary}) proposed DAGs almost identical to our Figure \ref{fig:dag_augmented} to describe the data generating mechanism behind the birth weight paradox. Specifically, the proposed DAGs can be written as special cases of the DAG in our Figure \ref{fig:dag_augmented} where $A$ is maternal smoking, $D_A$ is the low-birth-weight status that can be affected by smoking (``low birthweight from the predominant distribution'' or ``low birthweight due to smoking''\cite{vanderweele2014commentary}), $D$ is low-birth-weight status due to any reason, $Y$ is infant mortality, and $U$ represents unmeasured common causes of $D$ and $Y$ (e.g., ``malnutrition
or birth defects''\cite{vanderweele2014commentary}). The DAG allows maternal smoking ($A$) to induce a type of low-birth-weight status ($D_A$), while other unmeasured variables ($U$), that have an effect on mortality ($Y$), can also affect low birth weight ($D$). As Vanderweele\cite{vanderweele2014commentary} pointed out, Figure \ref{fig:dag_augmented} implies that $D_A$ is not associated with other unmeasured causes of $D$ or $Y$. Under this causal model, our Assumptions \ref{ass:swig_augmented}, \ref{ass:positivity} and \ref{ass:D_Asufficient_po} hold true. In particular, the determinism between $D_A$ and $D$ is implicit in the variables definition.

This means that, if the proposed causal model is true, the association among high-birth-weight babies ($D=0$) represents a valid estimate of the $CSE(a')$ if the treatment ``maternal smoking'' is decomposed into different components $A_D$ and $A_Y$ (Assumption \ref{ass:treatment_decomposition} and \ref{ass:swig_augmented_4arm}). 
However, it is unclear whether treatment decompositions are plausible to conceive for maternal smoking, as the same mechanisms proposed for the effect of smoking on birth weight\cite{ananth2004reexamining} also might affect mortality. If such decompositions are not of interest, questions about the direct effect of maternal smoking on infant mortality might be of little public health relevance. Under different assumptions, the smoking-mortality association among high-birth-weight infants may be viewed as a valid estimate of the $SACE$, as we show in Appendix \ref{app:sace}. However, the $SACE$ requires arguably strong justifications to be practically relevant, which are not clear to us in this example.

Our results show that it is possible to estimate a direct effect of maternal smoking on infant mortality, the $CSE(a')$, even when unmeasured confounders of birth weight and mortality exist. This is compatible with the observation that the birth weight paradox does not exist among high-birth-weight babies\cite{wilcox2006invited}: under our model, the association in high-birth-weight babies is a valid estimate of an effect of maternal smoking on mortality. Our results rely on a causal model already invoked in previous descriptions of the birth weight paradox. In this work, we take the causal model as given and do not examine whether it is plausible or compatible with the observed data. 

\section{Conclusion}
Under specific causal models, the association between treatment and outcome conditional on not having the post-treatment event represents a valid estimator of both conditional separable effects and survivor average causal effects, even when unmeasured common causes of the post-treatment event and the outcome exist. What we require is that the treatment and other unmeasured causes of the outcome generate the post-treatment event through ``independent mechanisms''. While this assumption needs to be evaluated on a case-by-case basis, it can be justified in settings where strong knowledge of biological mechanisms exists.

\section*{Acknowledgments}
Marco Piccininni and Mats J. Stensrud were supported by the Swiss National Science Foundation (SNSF Starting Grants, Grant number: 211550). The authors would like to thank the Isaac Newton Institute for Mathematical Sciences, Cambridge, for support and hospitality during the programme Causal inference: From theory to practice and back again, where work on this paper was undertaken. This work was supported by EPSRC grant EP/Z000580/1.

During the preparation of this manuscript the authors used ChatGPT in order to assist with formatting, minor language editing, and error checking. The authors reviewed and edited the content as needed and take full responsibility for the content of the publication.

\bibliographystyle{unsrt.bst}
\bibliography{references}

@article{cottone2023estimand,
  title={The estimand framework had implications in time to patient-reported outcomes deterioration analyses in cancer clinical trials},
  author={Cottone, Francesco and Efficace, Fabio and Cella, David and Aaronson, Neil K and Giesinger, Johannes M and Bachet, Jean-Baptiste and Louvet, Christophe and Charton, Emilie and Collins, Gary S and Anota, Amelie},
  journal={Journal of Clinical Epidemiology},
  volume={162},
  pages={118--126},
  year={2023},
  publisher={Elsevier}
}

@article{sadoff2021interim,
  title={Interim results of a phase 1--2a trial of Ad26. COV2. S Covid-19 vaccine},
  author={Sadoff, Jerald and Le Gars, Mathieu and Shukarev, Georgi and Heerwegh, Dirk and Truyers, Carla and de Groot, Anne M and Stoop, Jeroen and Tete, Sarah and Van Damme, Wim and Leroux-Roels, Isabel and others},
  journal={New England Journal of Medicine},
  volume={384},
  number={19},
  pages={1824--1835},
  year={2021},
  publisher={Mass Medical Soc}
}

@book{hernan2010causal,
  title={Causal Inference: What If},
  author={Hern{\'a}n, Miguel A and Robins, James M},
  year={2020},
  publisher={Boca Raton: Chapman \& Hall/CRC}
}

@article{hernan2004structural,
  title={A structural approach to selection bias},
  author={Hern{\'a}n, Miguel A and Hern{\'a}ndez-D{\'\i}az, Sonia and Robins, James M},
  journal={Epidemiology},
  volume={15},
  number={5},
  pages={615--625},
  year={2004},
  publisher={LWW}
}

@article{cines2021sars,
  title={SARS-CoV-2 vaccine--induced immune thrombotic thrombocytopenia},
  author={Cines, Douglas B and Bussel, James B},
  journal={New England Journal of Medicine},
  volume={384},
  number={23},
  pages={2254--2256},
  year={2021},
  publisher={Mass Medical Soc}
}

@article{greinacher2022vaccine,
  title={Vaccine-induced immune thrombotic thrombocytopenia (VITT): update on diagnosis and management considering different resources},
  author={Greinacher, Andreas and Langer, Florian and Makris, Mike and Pai, Menaka and Pavord, Sue and Tran, Huyen and Warkentin, Theodore E},
  journal={Journal of Thrombosis and Haemostasis},
  volume={20},
  number={1},
  pages={149--156},
  year={2022},
  publisher={Wiley Online Library}
}

@inproceedings{malinsky2019potential,
  title={A potential outcomes calculus for identifying conditional path-specific effects},
  author={Malinsky, Daniel and Shpitser, Ilya and Richardson, Thomas},
  booktitle={The 22nd International Conference on Artificial Intelligence and Statistics},
  pages={3080--3088},
  year={2019},
  organization={PMLR}
}

@article{holmberg2022collider,
  title={Collider bias},
  author={Holmberg, Mathias J and Andersen, Lars W},
  journal={Jama},
  volume={327},
  number={13},
  pages={1282--1283},
  year={2022}
}

@article{banack2024collider,
  title={Collider stratification bias I: principles and structure},
  author={Banack, Hailey R and Mayeda, Elizabeth Rose and Naimi, Ashley I and Fox, Matthew P and Whitcomb, Brian W},
  journal={American journal of epidemiology},
  volume={193},
  number={2},
  pages={238--240},
  year={2024},
  publisher={Oxford University Press}
}

@article{schneiderbanger2014management,
  title={Management of malignant hyperthermia: diagnosis and treatment},
  author={Schneiderbanger, Daniel and Johannsen, Stephan and Roewer, Norbert and Schuster, Frank},
  journal={Therapeutics and clinical risk management},
  pages={355--362},
  year={2014},
  publisher={Taylor \& Francis}
}

@article{hajek2022electronic,
  title={Electronic cigarettes versus nicotine patches for smoking cessation in pregnancy: a randomized controlled trial},
  author={Hajek, Peter and Przulj, Dunja and Pesola, Francesca and Griffiths, Chris and Walton, Robert and McRobbie, Hayden and Coleman, Tim and Lewis, Sarah and Whitemore, Rachel and Clark, Miranda and others},
  journal={Nature medicine},
  volume={28},
  number={5},
  pages={958--964},
  year={2022},
  publisher={Nature Publishing Group US New York}
}

@article{sallard2026novel,
  title={Novel adenovirus vaccine vectors lacking thrombosis-associated interactions with platelet factor 4},
  author={Sallard, Erwan and Pembaur, Daniel and Ciancaglini, Matias and Manov-Bouard, Lucie and Weklak, Denice and Hamdan, Firas and Chan, Chun Kit and J{\"o}nsson, Franziska and Chabot, Elise and Musielak, Carmen and others},
  journal={iScience},
  volume={29},
  number={1},
  year={2026},
  publisher={Elsevier}
}

@misc{angrist1995identification,
  title={Identification and estimation of local average treatment effects},
  author={Angrist, Joshua and Imbens, Guido},
  year={1995},
  publisher={National Bureau of Economic Research Cambridge, Mass., USA}
}

@article{richardson2013single,
  title={Single world intervention graphs (SWIGs): A unification of the counterfactual and graphical approaches to causality},
  author={Richardson, Thomas S and Robins, James M},
  journal={Center for the Statistics and the Social Sciences, University of Washington Series. Working Paper},
  volume={128},
  number={30},
  pages={2013},
  year={2013},
  publisher={Citeseer}
}

@article{young2021identified,
  title={Identified versus interesting causal effects in fertility trials and other settings with competing or truncation events},
  author={Young, Jessica G and Stensrud, Mats J},
  journal={Epidemiology},
  volume={32},
  number={4},
  pages={569--572},
  year={2021},
  publisher={LWW}
}

@article{stensrud2022translating,
  title={Translating questions to estimands in randomized clinical trials with intercurrent events},
  author={Stensrud, Mats J and Dukes, Oliver},
  journal={Statistics in Medicine},
  volume={41},
  number={16},
  pages={3211--3228},
  year={2022},
  publisher={Wiley Online Library}
}

@techreport{ich2019addendum,
  author       = {{ICH}},
  title        = {{ICH E9(R1): Addendum on Estimands and Sensitivity Analysis in Clinical Trials to the Guideline on Statistical Principles for Clinical Trials}},
  institution  = {ICH},
  year         = {2019}}

@article{rubin2006causal,
  title={Causal inference through potential outcomes and principal stratification: application to studies with" censoring" due to death},
  author={Rubin, Donald B},
  journal={Statistical Science},
  pages={299--309},
  year={2006},
  publisher={JSTOR}
}

@inbook{robins2022interventionist,
author = {Robins, James M. and Richardson, Thomas S. and Shpitser, Ilya},
title = {An Interventionist Approach to Mediation Analysis},
year = {2022},
isbn = {9781450395861},
publisher = {Association for Computing Machinery},
address = {New York, NY, USA},
edition = {1},
url = {https://doi.org/10.1145/3501714.3501754},
booktitle = {Probabilistic and Causal Inference: The Works of Judea Pearl},
pages = {713–764},
numpages = {52}
}

@article{robins2010alternative,
  title={Alternative graphical causal models and the identification of direct effects},
  author={Robins, James M and Richardson, Thomas S},
  journal={Causality and Psychopathology: Finding the Determinants of Disorders and their Cures},
  volume={84},
  pages={103--158},
  year={2010},
  publisher={Oxford University Press Oxford, UK}
}

@article{robins1992identifiability,
  title={Identifiability and exchangeability for direct and indirect effects},
  author={Robins, James M and Greenland, Sander},
  journal={Epidemiology},
  volume={3},
  number={2},
  pages={143--155},
  year={1992},
  publisher={LWW}
}

@article{park2024proximal,
  title={Proximal causal inference for conditional separable effects},
  author={Park, Chan and Stensrud, Mats and Tchetgen, Eric Tchetgen},
  journal={arXiv preprint arXiv:2402.11020},
  year={2024}
}

@article{robins1986new,
  title={A new approach to causal inference in mortality studies with a sustained exposure period—application to control of the healthy worker survivor effect},
  author={Robins, James},
  journal={Mathematical Modelling},
  volume={7},
  number={9-12},
  pages={1393--1512},
  year={1986},
  publisher={Elsevier}
}

@article{stensrud2021generalized,
  title={A generalized theory of separable effects in competing event settings},
  author={Stensrud, Mats J and Hern{\'a}n, Miguel A and Tchetgen Tchetgen, Eric J and Robins, James M and Didelez, Vanessa and Young, Jessica G},
  journal={Lifetime Data Analysis},
  volume={27},
  number={4},
  pages={588--631},
  year={2021},
  publisher={Springer}
}

@article{stensrud2022separable,
  title={Separable effects for causal inference in the presence of competing events},
  author={Stensrud, Mats J and Young, Jessica G and Didelez, Vanessa and Robins, James M and Hern{\'a}n, Miguel A},
  journal={Journal of the American Statistical Association},
  volume={117},
  number={537},
  pages={175--183},
  year={2022},
  publisher={Taylor \& Francis}
}

@article{stensrud2023conditional,
  title={Conditional separable effects},
  author={Stensrud, Mats J and Robins, James M and Sarvet, Aaron and Tchetgen Tchetgen, Eric J and Young, Jessica G},
  journal={Journal of the American Statistical Association},
  volume={118},
  number={544},
  pages={2671--2683},
  year={2023},
  publisher={Taylor \& Francis}
}

@article{wang2017identification,
  title={Identification and estimation of causal effects with outcomes truncated by death},
  author={Wang, Linbo and Zhou, Xiao-Hua and Richardson, Thomas S},
  journal={Biometrika},
  volume={104},
  number={3},
  pages={597--612},
  year={2017},
  publisher={Oxford University Press}
}

@article{wang2026adenoviral,
  title={Adenoviral Inciting Antigen and Somatic Hypermutation in VITT},
  author={Wang, Jing Jing and Sch{\"o}nborn, Linda and Warkentin, Theodore E and M{\"u}ller, Luisa and Thiele, Thomas and Ulm, Lena and V{\"o}lker, Uwe and Ameling, Sabine and Franzenburg, S{\"o}ren and Kaderali, Lars and others},
  journal={New England Journal of Medicine},
  volume={394},
  number={7},
  pages={669--683},
  year={2026},
  publisher={Mass Medical Soc}
}

@article{yerushalmy1971relationship,
  title={THE RELATIONSHIP OF PARENTS CIGARETTE SMOKING TO OUTCOME OF PREGNANCY--IMPLICATIONS AS TO THE PROBLEM OF INFERRING CAUSATION FROM OBSERVED ASSOCIATIONS},
  author={Yerushalmy, J},
  journal={American Journal of Epidemiology},
  volume={93},
  number={6},
  pages={443--443},
  year={1971},
  publisher={Oxford University Press}
}

@article{wilcox2001importance,
  title={On the importance—and the unimportance—of birthweight},
  author={Wilcox, Allen J},
  journal={International journal of epidemiology},
  volume={30},
  number={6},
  pages={1233--1241},
  year={2001},
  publisher={Oxford University Press}
}

@article{joseph2004parsimonious,
  title={A parsimonious explanation for intersecting perinatal mortality curves: understanding the effects of race and of maternal smoking},
  author={Joseph, KS and Demissie, Kitaw and Platt, Robert W and Ananth, Cande V and McCarthy, Brian J and Kramer, Michael S},
  journal={BMC Pregnancy and Childbirth},
  volume={4},
  number={1},
  pages={7},
  year={2004},
  publisher={Springer}
}

@article{hernandez2006birth,
  title={The birth weight “paradox” uncovered?},
  author={Hern{\'a}ndez-D{\'\i}az, Sonia and Schisterman, Enrique F and Hern{\'a}n, Miguel A},
  journal={American journal of epidemiology},
  volume={164},
  number={11},
  pages={1115--1120},
  year={2006},
  publisher={Oxford University Press}
}

@article{ananth2004reexamining,
  title={Reexamining the effects of gestational age, fetal growth, and maternal smoking on neonatal mortality},
  author={Ananth, Cande V and Platt, Robert W},
  journal={BMC pregnancy and childbirth},
  volume={4},
  number={1},
  pages={22},
  year={2004},
  publisher={Springer}
}

@book{pearl2018book,
  title={The book of why: the new science of cause and effect},
  author={Pearl, Judea and Mackenzie, Dana},
  year={2018},
  publisher={Basic books}
}

@article{schisterman2009z,
  title={Z-scores and the birthweight paradox},
  author={Schisterman, Enrique F and Whitcomb, Brian W and Mumford, Sunni L and Platt, Robert W},
  journal={Paediatric and Perinatal Epidemiology},
  volume={23},
  number={5},
  pages={403--413},
  year={2009},
  publisher={Wiley Online Library}
}

@article{whitcomb2009quantification,
  title={Quantification of collider-stratification bias and the birthweight paradox},
  author={Whitcomb, Brian W and Schisterman, Enrique F and Perkins, Neil J and Platt, Robert W},
  journal={Paediatric and perinatal epidemiology},
  volume={23},
  number={5},
  pages={394--402},
  year={2009},
  publisher={Wiley Online Library}
}

@article{vanderweele2014commentary,
  title={Commentary: Resolutions of the birthweight paradox: competing explanations and analytical insights},
  author={VanderWeele, Tyler J},
  journal={International journal of epidemiology},
  volume={43},
  number={5},
  pages={1368--1373},
  year={2014},
  publisher={Oxford University Press}
}

@article{wilcox2006invited,
  title={Invited commentary: the perils of birth weight—a lesson from directed acyclic graphs},
  author={Wilcox, Allen J},
  journal={American Journal of Epidemiology},
  volume={164},
  number={11},
  pages={1121--1123},
  year={2006},
  publisher={Oxford University Press}
}

\newpage

\begin{appendices}
\setcounter{page}{1}
\begin{center}
\textbf{APPENDIX - EFFECTS CONDITIONAL ON POST-TREATMENT EVENTS GENERATED BY INDEPENDENT MECHANISMS}

MARCO PICCININNI, MATS J. STENSRUD

\end{center}

\newpage

\section{Proofs of main results} \label{app:proofs}

\begin{proof}[Proof of Lemma \ref{lemma:multiplicative_survival}]
\begin{align*}
&\Pr(D^{a=a'}=0 \mid U=u)\\
&= \Pr(D^{a=a'}=0 \mid U=u, D^{a=a'}_A=0) \Pr(D^{a=a'}_A=0 \mid U=u)\\
&\quad + \Pr(D^{a=a'}=0 \mid U=u, D^{a=a'}_A=1) \Pr(D^{a=a'}_A=1 \mid U=u)\\
&\overset{(\text{A\ref{ass:D_Asufficient_po}})}{=} \Pr(D^{a=a'}=0 \mid U=u, D^{a=a'}_A=0) \Pr(D^{a=a'}_A=0 \mid U=u)\\
&\overset{(\text{A\ref{ass:swig_augmented}})}{=} \Pr(D=0 \mid U=u, D_A=0) \Pr(D_A=0 \mid A=a').
\end{align*}
Assumption \ref{ass:positivity} ensures that the conditional quantities are well-defined.
\end{proof}

\begin{proof}[Proof of Lemma \ref{lemma:useful_factorization}]
For every $a' \in \{0,1\}$,
\begin{align*}
&\Pr(U=u \mid D^{a=a'}=0)=\frac{\Pr(D^{a=a'}=0 \mid U=u) \Pr(U=u)}{\sum_v \Pr(D^{a=a'}=0 \mid U=v) \Pr(U=v)}\\
&\overset{(\text{L\ref{lemma:multiplicative_survival}})}{=}\frac{\Pr(D=0 \mid U=u, D_A=0) \Pr(D_A=0 \mid A=a') \Pr(U=u)}{\sum_v \Pr(D=0 \mid U=v, D_A=0) \Pr(D_A=0 \mid A=a') \Pr(U=v)}\\
&=\frac{\Pr(D=0 \mid U=u, D_A=0) \Pr(U=u)}{\sum_v \Pr(D=0 \mid U=v, D_A=0) \Pr(U=v)},
\end{align*}    
where the final term is not a function of $a'$. Assumption \ref{ass:positivity} ensures that the conditional quantities are well-defined.
\end{proof}

\begin{proof}[Proof of Proposition \ref{prop:cse_identification}]
First note that in the four-arm trial, under Assumptions \ref{ass:D_Asufficient_po}, \ref{ass:treatment_decomposition}, and \ref{ass:swig_augmented_4arm}, we have that
\begin{align}\label{eq:lemma2_sepeff}
\Pr(U=u \mid D^{a_D=0,a_Y=0}=0)=\Pr(U=u \mid D^{a_D=1,a_Y=1}=0). \tag{*}   
\end{align}
This can be proven by following the same steps of the proofs for Lemma \ref{lemma:multiplicative_survival} and Lemma \ref{lemma:useful_factorization} but replacing interventions or conditioning on $A=a'$ with interventions or conditioning on $A_D=A_Y=a'$ (Assumption \ref{ass:treatment_decomposition}) and invoking Assumption \ref{ass:swig_augmented_4arm} instead of Assumption \ref{ass:swig_augmented}.

Then, for every $a'$,
\begin{align}\label{eq:determinism_D_DA}
&D^{a_D=a'}=0 \overset{(\text{A\ref{ass:swig_augmented_4arm}})}{\iff} D^{a_D=a',a_Y=a'}=0 \overset{(\text{A\ref{ass:treatment_decomposition}})}{\iff} D^{a=a'}=0 \notag \\
&\overset{(\text{A\ref{ass:D_Asufficient_po}})}{\implies} D^{a=a'}_A=0 \overset{(\text{A\ref{ass:treatment_decomposition}})}{\iff} D^{a_D=a', a_Y=a'}_A=0 \overset{(\text{A\ref{ass:swig_augmented_4arm}})}{\iff} D^{a_D=a'}_A=0. \tag{**}
\end{align}
For any $a,a' \in \{0,1\}$, the quantity
\begin{align*}
&\mathbb{E}(Y^{a_D=a', a_Y=a} \mid D^{a_D=a'}=0)\\
&=\sum_u \Bigl\{ \mathbb{E}(Y^{a_D=a', a_Y=a} \mid D^{a_D=a'}=0, U=u)\Pr(U=u \mid D^{a_D=a'}=0)\Bigr\}\\
&\overset{(\text{A\ref{ass:swig_augmented_4arm}})}{=}\sum_u \Bigl\{ \mathbb{E}(Y^{a_D=a', a_Y=a} \mid D^{a_D=a'}=0, U=u)\Pr(U=u \mid D^{a_D=a',a_Y=a'}=0)\Bigr\}\\
&\overset{(\text{\ref{eq:lemma2_sepeff}})}{=}\sum_u \Bigl\{ \mathbb{E}(Y^{a_D=a', a_Y=a} \mid D^{a_D=a'}=0, U=u)\Pr(U=u \mid D^{a_D=1,a_Y=1}=0)\Bigr\}\\
&\overset{(\text{\ref{eq:determinism_D_DA}})}{=}\sum_u \Bigl\{ \mathbb{E}(Y^{a_D=a', a_Y=a} \mid D^{a_D=a'}=0, D^{a_D=a'}_A=0, U=u)\Pr(U=u \mid D^{a_D=1,a_Y=1}=0)\Bigr\}\\
&\overset{(\text{A\ref{ass:swig_augmented_4arm}})}{=}\sum_u \Bigl\{ \mathbb{E}(Y \mid D=0, D_A=0, U=u, A_Y=a)\Pr(U=u \mid D^{a_D=1,a_Y=1}=0)\Bigr\}
\end{align*}
is not a function of $a'$. From which follows that
\begin{align*}
\mathbb{E}(Y^{a_D=1, a_Y=a} \mid D^{a_D=1}=0)=
\mathbb{E}(Y^{a_D=0, a_Y=a} \mid D^{a_D=0}=0).
\end{align*}
Thus, with data from only the two arms $A_D=A_Y=1$ and $A_D=A_Y=0$, we can identify the expectation of the outcome conditional on $D=0$ in all arms of the hypothetical four-arm trial. That is,
\begin{align*}
&\mathbb{E}(Y^{a_D=0, a_Y=1} \mid D^{a_D=0}=0) =
\mathbb{E}(Y^{a_D=1, a_Y=1} \mid D^{a_D=1}=0)\\
&\overset{(\text{A\ref{ass:swig_augmented_4arm}})}{=}
\mathbb{E}(Y \mid A_D=1, A_Y=1, D=0)\overset{(\text{A\ref{ass:treatment_decomposition}})}{=} \mathbb{E}(Y \mid A=1, D=0),
\end{align*}
and
\begin{align*}
&\mathbb{E}(Y^{a_D=1, a_Y=0} \mid D^{a_D=1}=0) =
\mathbb{E}(Y^{a_D=0, a_Y=0} \mid D^{a_D=0}=0)\\
&\overset{(\text{A\ref{ass:swig_augmented_4arm}})}{=}
\mathbb{E}(Y \mid A_D=0, A_Y=0, D=0)\overset{(\text{A\ref{ass:treatment_decomposition}})}{=} \mathbb{E}(Y \mid A=0, D=0).
\end{align*}  
From which the result in Proposition \ref{prop:cse_identification} follows immediately. Assumption \ref{ass:positivity} ensures that the conditional quantities are well-defined.
\end{proof}

\newpage

\section{Sufficient-component cause model} \label{app:causal_pies}
Here we give a simple, illustrative example of a sufficient-component cause model\cite{hernan2010causal}. This is meant to give readers a conceptual understanding of independent mechanisms. Each ``causal pie'' represents a sufficient cause for the post-treatment event, which is binary. Consider $\delta_0, \delta_1, \delta_2, \delta_3, \delta_4, \delta_5, A, U \in \{0,1\}$ to be mutually independent binary variables. As in the main text, $A$ represents the randomized treatment ($A=1$ treatment of interest, $A=0$ control), while $U$ represents all unmeasured causes of $D$ that are also associated with $Y$. Here, for simplicity we assume $U$ to be binary. The $\delta$-s represent some exogenous (independent ``noise'') variables.

Consider the following causal pies representing all possible sufficient causes for the realization of the event $D=1$.

\begin{figure}[htbp]
  \centering
  \includegraphics[width=0.7\textwidth]{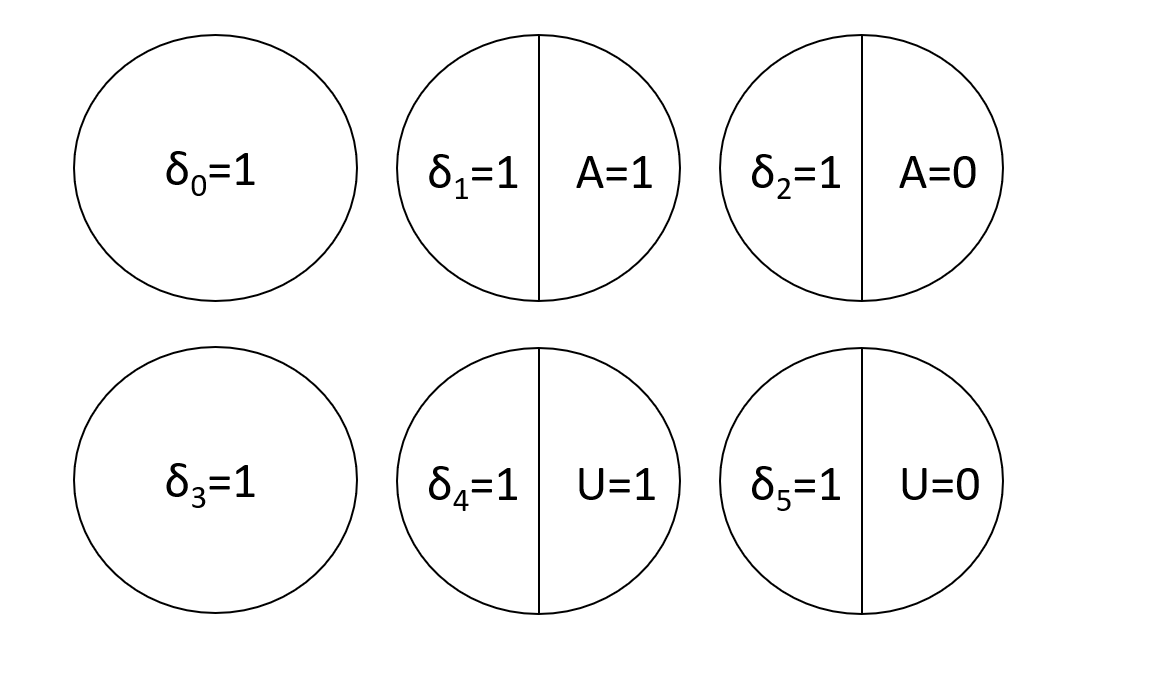}
  \caption*{}
\end{figure}

Define $D_A$ to be a binary variable that is equal to 1 when at least one of the first row of causal pies hold true, that is, $$D_A:=(\delta_0=1) \vee (\delta_1=1 \wedge A=1) \vee (\delta_2=1 \wedge A=0),$$
where $\vee$ and $\wedge$ indicate the OR and AND logical operators, respectively. 

The variable $D_A$ is a sufficient cause for $D$ (the realization of $D_A=1$ ensures that $D=1$), fully mediates the effect of $A$ on $D$ (no other pie includes $A$ outside the ones concurring to $D_A$), and is independent of other unmeasured causes of $D$ (because $D_A$ is a deterministic function of variables jointly independent of $U,\delta_3,\delta_4,\delta_5$). This is analogous to the assumptions of the causal model in Figure \ref{fig:dag_augmented} and the determinism in Assumption \ref{ass:D_Asufficient_po}.

The set of causal pies we are considering ensures that $D=1$ happens through two distinct separate mechanisms: (1) $D_A$ is equal to 1, and (2) either $\delta_3=1$ or certain combinations of $U,\delta_4,\delta_5$ occur. There is no causal pie that contains both $A$ and $U$. The variables $A$ and $U$ do not interact, in the ``sufficient cause interaction'' sense\cite{hernan2010causal}, in generating $D=1$.

Under this set of causal pies, it is easy to show that the multiplicative survival model holds:

\begin{align*}
&\Pr(D^{a=a'}=0 \mid U=u)=\Pr(D=0 \mid A=a', U=u )\\
&=\begin{cases}
\Pr(\neg \left((D_A=1) \vee (\delta_3=1) \vee (\delta_4=1)\right) \mid A=a') & \text{if $u=1$}\\
\Pr(\neg \left((D_A=1) \vee (\delta_3=1) \vee (\delta_5=1)\right) \mid A=a') & \text{if $u=0$}
\end{cases}\\
&=\begin{cases}
\Pr((D_A=0) \wedge (\delta_3=0) \wedge (\delta_4=0) \mid A=a') & \text{if $u=1$}\\
\Pr( (D_A=0) \wedge (\delta_3=0) \wedge (\delta_5=0) \mid A=a') & \text{if $u=0$}
\end{cases}\\
&=\begin{cases}
\Pr(D_A=0 \mid A=a') \Pr((\delta_3=0) \wedge (\delta_4=0) \mid A=a') & \text{if $u=1$}\\
\Pr(D_A=0 \mid A=a') \Pr((\delta_3=0) \wedge (\delta_5=0) \mid A=a') & \text{if $u=0$}\\
\end{cases}\\
&=\begin{cases}
\Pr(D_A=0 \mid A=a')\Pr(D=0 \mid U=u, D_A=0) & \text{if $u=1$}\\
\Pr(D_A=0 \mid A=a') \Pr(D=0 \mid U=u, D_A=0) & \text{if $u=0$}.\\
\end{cases}
\end{align*}

\newpage

\section{Survivor average causal effect} \label{app:sace}

Consider the following monotonicity assumption:
\begin{assumption} \label{ass:monotonicityD_A}
$D_A^{a=1} \geq D_A^{a=0}.$    
\end{assumption}
In our plastic surgery example, Assumption \ref{ass:monotonicityD_A} can be interpreted as follows: no individual would die of $D_A$ when untreated but not die of $D_A$ when treated. This monotonicity assumption on $D_A$, immediately induces a monotonicity of treatment effects on $D$ as well, given the determinism and the exclusion restrictions in our causal model.

\begin{lemma} \label{lemma:monotonicity_D}
Under Assumptions \ref{ass:swig_augmented}, \ref{ass:D_Asufficient_po} and \ref{ass:monotonicityD_A}, we have that $$D^{a=0}=D^{a=1}=0 \iff D^{a=1}=0.$$
\end{lemma}

\begin{proof}[Proof of Lemma \ref{lemma:monotonicity_D}]
The direction $D^{a=0}=D^{a=1}=0 \implies D^{a=1}=0$ holds trivially. To show that the converse also holds, we have to show that $D^{a=1}=0 \implies D^{a=0}=0$. We will use that $D^{a=1}=0 \overset{(\text{A\ref{ass:D_Asufficient_po}})}{\implies} D^{a=1}_A=0 \overset{(\text{A\ref{ass:monotonicityD_A}})}{\implies} D^{a=0}_A=0$. Moreover, under Assumption \ref{ass:swig_augmented}, $D^{a=a'}=D^{a=a',d_A=D^{a=a'}_A}=D^{d_A=D^{a=a'}_A}$, therefore units with $D^{a=1}_A=D^{a=0}_A$ also have $D^{a=1}=D^{a=0}$. This proves that $D^{a=1}=0 \implies D^{a=0}=0$.
\end{proof}

Consider also an independence assumption conditional on $U$, which might be uncontroversial in a Non-Parametric Structural Equation Model with independent errors\cite{pearl2018book}, but not under a FFRCISTG model that does not invoke cross-world assumptions\cite{robins1986new,richardson2013single}: 

\begin{assumption} \label{ass:crossworld_ind}
$Y^{a=0} \independent D^{a=1} \mid D^{a=0}=0,U.$    
\end{assumption}
Under these additional assumptions, the $SACE$ can be identified with the same functional of Proposition \ref{prop:cse_identification}.

\begin{proposition}\label{prop:sace_identification}
Under Assumptions \ref{ass:swig_augmented}, \ref{ass:positivity}, \ref{ass:D_Asufficient_po}, \ref{ass:monotonicityD_A}, and \ref{ass:crossworld_ind}, 
$$SACE=\mathbb{E}(Y \mid A=1, D=0)-\mathbb{E}(Y \mid A=0, D=0).$$
\end{proposition}

\begin{proof}[Proof of Proposition \ref{prop:sace_identification}]
We have
\begin{align*}
&\mathbb{E}(Y^{a=0} \mid D^{a=0}=0,D^{a=1}=0)\\
&=
\sum_u \Bigl\{ \mathbb{E}(Y^{a=0} \mid D^{a=0}=0,D^{a=1}=0, U=u)\Pr(U=u \mid D^{a=0}=0,D^{a=1}=0)\Bigr\}\\
&\overset{(\text{L\ref{lemma:monotonicity_D}})}{=}
\sum_u \Bigl\{ \mathbb{E}(Y^{a=0} \mid D^{a=0}=0,D^{a=1}=0, U=u)\Pr(U=u \mid D^{a=1}=0)\Bigr\}\\
&\overset{(\text{A\ref{ass:crossworld_ind}})}{=}
\sum_u \Bigl\{ \mathbb{E}(Y^{a=0} \mid D^{a=0}=0, U=u)\Pr(U=u \mid D^{a=1}=0)\Bigr\}\\
&\overset{(\text{L\ref{lemma:useful_factorization}})}{=}\sum_u \Bigl\{ \mathbb{E}(Y^{a=0} \mid D^{a=0}=0, U=u)\Pr(U=u \mid D^{a=0}=0)\Bigr\}=\mathbb{E}(Y^{a=0} \mid D^{a=0}=0)
\end{align*}
Therefore,
\begin{align*}
&\mathbb{E}(Y^{a=1} \mid D^{a=0}=0, D^{a=1}=0)
\overset{(\text{L\ref{lemma:monotonicity_D}})}{=}
\mathbb{E}(Y^{a=1} \mid D^{a=1}=0)
\overset{(\text{A\ref{ass:swig_augmented}})}{=}
\mathbb{E}(Y \mid A=1, D=0)
\end{align*}
and
\begin{align*}
&\mathbb{E}(Y^{a=0} \mid D^{a=0}=0, D^{a=1}=0)=\mathbb{E}(Y^{a=0} \mid D^{a=0}=0)
\overset{(\text{A\ref{ass:swig_augmented}})}{=}
\mathbb{E}(Y \mid A=0, D=0).
\end{align*}
From which the result in Proposition \ref{prop:sace_identification} follows immediately. Assumption \ref{ass:positivity} ensures that the conditional quantities are well-defined.
\end{proof}

\newpage

\section{Identification when $L$ exists} \label{app:proof_identification_L}

Here we consider an identification result that is more general than the one presented in the main text. We consider the existence of a (measured) discrete variable $L$ that causes $D_A,D,$ and $Y$, and is affected by $A_D$.

\begin{assumption} \label{ass:swig_augmented_4arm_L}
For the hypothetical four-arm trial where both $A_D$ and $A_Y$ are randomized independently, assume a FFRCISTG model associated with the graph in Figure \ref{fig:dag_augmented_4arm_L}.
\end{assumption}

It follows from Assumption \ref{ass:swig_augmented_4arm_L} that the distribution of counterfactual variables factorizes according to the SWIG in Figure \ref{fig:swig_augmented_4arm_L}  and modularity holds under all possible interventions $a,a'\in\{0,1\}$.

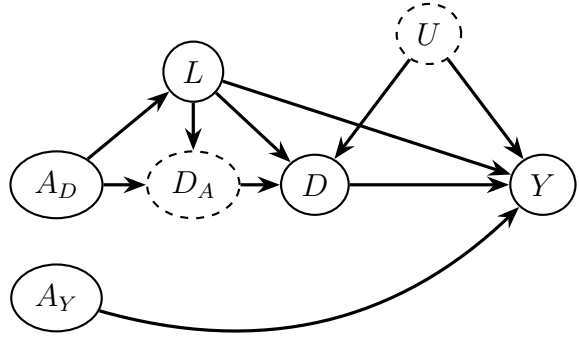
\begin{figure}[!ht] 
\centering
    \begin{tikzpicture}
       \begin{scope}[every node/.style={thick,draw}]
       \node[name=A_D, shape=ellipse] at (-3.4,0){$A_D$};
       \node[name=A_Y, shape=ellipse] at (-3.4,-1.5){$A_Y$};
       \node[name=L, shape=ellipse] at (-1.6,1.5){$L$};
       \node[name=D, shape=ellipse] at (0,0){$D$};
       \node[name=D_A, shape=ellipse, dashed] at (-1.6,0){$D_A$};
       \node[name=Y, shape=ellipse] at (3,0){$Y$};
       \node[name=U, shape=ellipse, dashed] at (1.5,2){$U$};
    \end{scope}
        
    \begin{scope}[>={Stealth[black]},
                  every node/.style={fill=white,circle},
                  every edge/.style={draw=black,very thick}]
        \path[->] (U) edge (Y);
        \path[->] (U) edge (D);
        \path[->] (A_D) edge (D_A);
        \path[->] (A_D) edge (L);
        \path[->] (D_A) edge (D);
        \path[->] (D_A) edge[bend right=25] (Y);
        \path[->] (D) edge (Y);
        \path[->] (A_Y) edge[bend right] (Y);
        \path[->] (L) edge (D_A);
        \path[->] (L) edge (D);
        \path[->] (L) edge (Y);
    \end{scope}
    \end{tikzpicture}
    \caption{Causal directed acyclic graph representing a hypothetical trial in which $A_D$ and $A_Y$ are randomized independently. Dashed nodes represent unmeasured variables. Here we consider also a variable $L$ which causes $D_A,D,$ and $Y$, and is affected by $A_D$.}
    \label{fig:dag_augmented_4arm_L}
\end{figure}

\begin{figure}[!ht] 
\centering
    \begin{tikzpicture}
       \begin{scope}[every node/.style={thick,draw}]
       \node[name=A_D,shape=swig vsplit] at (-5.8,0) {
                                              \nodepart{left}{$A_D$}
                                              \nodepart{right}{$a'$} };
        \node[name=A_Y,shape=swig vsplit] at (-5.8,-1.5) {
                                              \nodepart{left}{$A_Y$}
                                              \nodepart{right}{$a$} };
       \node[name=D, shape=ellipse] at (0,0){$D^{a_D=a'}$};
       \node[name=L, shape=ellipse] at (-2.75,1.8){$L^{a_D=a'}$};
       \node[name=D_A, shape=ellipse, dashed] at (-2.75,0){$D^{a_D=a'}_A$};
       \node[name=Y, shape=ellipse] at (4,0){$Y^{a_D=a',\ a_Y=a}$};
       \node[name=U, shape=ellipse, dashed] at (2,2){$U$};
    \end{scope}
        
    \begin{scope}[>={Stealth[black]},
                  every node/.style={fill=white,circle},
                  every edge/.style={draw=black,very thick}]
        \path[->] (U) edge (Y);
        \path[->] (U) edge (D);
        \path[->] (A_D) edge (D_A);
        \path[->] (D_A) edge (D);
        \path[->] (D_A) edge[bend right=20] (Y);
        \path[->] (D) edge (Y);
        \path[->] (A_Y) edge[bend right] (Y);
        \path[->] (A_D) edge (L);
        \path[->] (L) edge (D_A);
        \path[->] (L) edge (D);
        \path[->] (L) edge (Y);
    \end{scope}
    \end{tikzpicture}
    \caption{Single world intervention graph  corresponding to the directed acyclic graph in Figure \ref{fig:dag_augmented_4arm_L} representing the intervention on the components of treatment $A_D$ and $A_Y$ in the hypothetical four-arm trial when a variable $L$ exists. Dashed nodes represent unmeasured variables.}
    \label{fig:swig_augmented_4arm_L}
\end{figure}
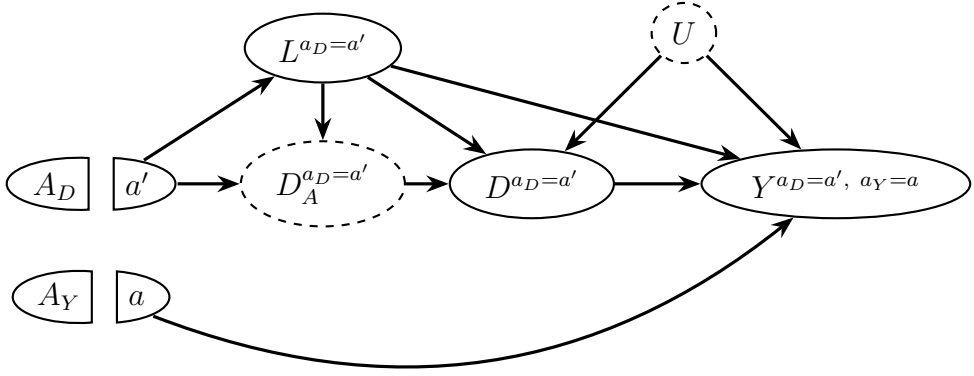

We will also make the following positivity assumption.
\begin{assumption} \label{ass:positivity_L}
If $\Pr(D=0,L=l)>0 \implies \Pr(D=0,L=l \mid A=a)>0$ for every $a$ and $l$. 
\end{assumption}

\begin{proposition} \label{prop:cse_identification_L}
Under Assumptions \ref{ass:positivity}, \ref{ass:D_Asufficient_po}, \ref{ass:treatment_decomposition}, \ref{ass:swig_augmented_4arm_L}, and \ref{ass:positivity_L},
\begin{align*}
&CSE(a')\\
&=\sum_l \{\mathbb{E}(Y \mid D=0, L=l, A=1)- \mathbb{E}(Y \mid D=0, L=l, A=0)\} \Pr(L=l \mid D=0, A=a').
\end{align*}
\end{proposition}

\begin{proof}[Proof of Proposition \ref{prop:cse_identification_L}]
Following a similar rationale to the proof of Lemma \ref{lemma:multiplicative_survival}, we can show that
\begin{align*}
&\Pr(D^{a_D=a'}=0 \mid U=u, L^{a_D=a'}=l)\\
&= \Pr(D^{a_D=a'}=0 \mid U=u, L^{a_D=a'}=l, D^{a_D=a'}_A=0) \Pr(D^{a_D=a'}_A=0 \mid U=u, L^{a_D=a'}=l)\\
&\quad + \Pr(D^{a_D=a'}=0 \mid U=u, L^{a_D=a'}=l, D^{a_D=a'}_A=1) \Pr(D^{a_D=a'}_A=1 \mid U=u, L^{a_D=a'}=l)\\
&\overset{(\text{A\ref{ass:swig_augmented_4arm_L}})}{=} \Pr(D^{a_D=a'}=0 \mid U=u, L^{a_D=a'}=l, D^{a_D=a'}_A=0) \Pr(D^{a_D=a'}_A=0 \mid U=u, L^{a_D=a'}=l)\\
&\quad + \Pr(D^{a_D=a', a_Y=a'}=0 \mid U=u, L^{a_D=a', a_Y=a'}=l, D^{a_D=a', a_Y=a'}_A=1)\\
&\quad \quad \Pr(D^{a_D=a', a_Y=a'}_A=1 \mid U=u, L^{a_D=a', a_Y=a'}=l)\\
&\overset{(\text{A\ref{ass:treatment_decomposition}})}{=} \Pr(D^{a_D=a'}=0 \mid U=u, L^{a_D=a'}=l, D^{a_D=a'}_A=0) \Pr(D^{a_D=a'}_A=0 \mid U=u, L^{a_D=a'}=l)\\
&\quad + \Pr(D^{a=a'}=0 \mid U=u, L^{a=a'}=l, D^{a=a'}_A=1) \Pr(D^{a=a'}_A=1 \mid U=u, L^{a=a'}=l)\\
&\overset{(\text{A\ref{ass:D_Asufficient_po}})}{=} \Pr(D^{a_D=a'}=0 \mid U=u, L^{a_D=a'}=l, D^{a_D=a'}_A=0) \Pr(D^{a_D=a'}_A=0 \mid U=u, L^{a_D=a'}=l)\\
&\overset{(\text{A\ref{ass:swig_augmented_4arm_L}})}{=} \Pr(D=0 \mid U=u, L=l, D_A=0) \Pr(D_A=0 \mid L=l, A_D=a').
\end{align*}
Following an argument similar to the one in the proof of Lemma \ref{lemma:useful_factorization},
\begin{align*}
&\Pr(U=u \mid D^{a_D=a'}=0, L^{a_D=a'}=l)\\
&=\frac{\Pr(D^{a_D=a'}=0 \mid U=u, L^{a_D=a'}=l) \Pr(U=u \mid L^{a_D=a'}=l)}{\sum_v \Pr(D^{a_D=a'}=0 \mid U=v, L^{a_D=a'}=l) \Pr(U=v \mid L^{a_D=a'}=l)}\\
&=\frac{\Pr(D=0 \mid U=u, L=l, D_A=0) \Pr(D_A=0 \mid L=l, A_D=a') \Pr(U=u \mid L^{a_D=a'}=l)}{\sum_v \Pr(D=0 \mid U=v, L=l, D_A=0) \Pr(D_A=0 \mid L=l, A_D=a') \Pr(U=v \mid L^{a_D=a'}=l)}\\
&=\frac{\Pr(D=0 \mid U=u, L=l, D_A=0) \Pr(U=u \mid L^{a_D=a'}=l)}{\sum_v \Pr(D=0 \mid U=v, L=l, D_A=0) \Pr(U=v \mid L^{a_D=a'}=l)}\\
&\overset{(\text{A\ref{ass:swig_augmented_4arm_L}})}{=}\frac{\Pr(D=0 \mid U=u, L=l, D_A=0) \Pr(U=u)}{\sum_v \Pr(D=0 \mid U=v, L=l, D_A=0) \Pr(U=v)}.
\end{align*}

This shows that the probability of $U$ conditional on $D^{a_D=a'}=0,L^{a_D=a'}$ does not depend on the value of $a'$
\begin{align}\label{eq:lemma2_sepeff_withL}
\Pr(U=u \mid D^{a_D=1}=0, L^{a_D=1}=l)=\Pr(U=u \mid D^{a_D=0}=0, L^{a_D=0}=l).\tag{$\star$}
\end{align}

Furthermore, for every $a'$,
\begin{align}\label{eq:determinism_D_DA_withL}
&D^{a_D=a'}=0 \overset{(\text{A\ref{ass:swig_augmented_4arm_L}})}{\iff} D^{a_D=a',a_Y=a'}=0 \overset{(\text{A\ref{ass:treatment_decomposition}})}{\iff} D^{a=a'}=0 \notag \\
&\overset{(\text{A\ref{ass:D_Asufficient_po}})}{\implies} D^{a=a'}_A=0 \overset{(\text{A\ref{ass:treatment_decomposition}})}{\iff} D^{a_D=a', a_Y=a'}_A=0 \overset{(\text{A\ref{ass:swig_augmented_4arm_L}})}{\iff} D^{a_D=a'}_A=0. \tag{$\star \star$}
\end{align}
For any $a,a' \in \{0,1\}$, the quantity
\begin{align*}
&\mathbb{E}(Y^{a_D=a', a_Y=a} \mid D^{a_D=a'}=0, L^{a_D=a'}=l)\\
&=
\sum_u \Bigl\{ \mathbb{E}(Y^{a_D=a', a_Y=a} \mid D^{a_D=a'}=0, L^{a_D=a'}=l, U=u)\Pr(U=u \mid D^{a_D=a'}=0, L^{a_D=a'}=l)\Bigr\}\\
&\overset{(\text{\ref{eq:lemma2_sepeff_withL}})}{=}
\sum_u \Bigl\{ \mathbb{E}(Y^{a_D=a', a_Y=a} \mid D^{a_D=a'}=0, L^{a_D=a'}=l, U=u)\Pr(U=u \mid D^{a_D=1}=0, L^{a_D=1}=l)\Bigr\}\\
&\overset{(\text{\ref{eq:determinism_D_DA_withL}})}{=}
\sum_u \Bigl\{ \mathbb{E}(Y^{a_D=a', a_Y=a} \mid D^{a_D=a'}=0, D^{a_D=a'}_A=0, L^{a_D=a'}=l, U=u)\\
&\quad \Pr(U=u \mid D^{a_D=1}=0, L^{a_D=1}=l)\Bigr\}\\
&\overset{(\text{A\ref{ass:swig_augmented_4arm_L}})}{=}\sum_u \Bigl\{ \mathbb{E}(Y \mid D=0, D_A=0, L=l, U=u, A_Y=a)\Pr(U=u \mid D^{a_D=1}=0, L^{a_D=1}=l)\Bigr\}
\end{align*}
is not a function of $a'$. From which follows that
\begin{align*}
\mathbb{E}(Y^{a_D=1, a_Y=a} \mid D^{a_D=1}=0, L^{a_D=1}=l)=\mathbb{E}(Y^{a_D=0, a_Y=a} \mid D^{a_D=0}=0, L^{a_D=0}=l).
\end{align*}
Thus, we can easily identify the following quantities:
\begin{align*}
&\mathbb{E}(Y^{a_D=0, a_Y=1} \mid D^{a_D=0}=0, L^{a_D=0}=l) =
\mathbb{E}(Y^{a_D=1, a_Y=1} \mid D^{a_D=1}=0, L^{a_D=1}=l)\\
&\overset{(\text{A\ref{ass:swig_augmented_4arm_L}})}{=}
\mathbb{E}(Y \mid D=0, L=l, A_D=1, A_Y=1)\overset{(\text{A\ref{ass:treatment_decomposition}})}{=} \mathbb{E}(Y \mid D=0, L=l, A=1),
\end{align*}
and
\begin{align*}
&\mathbb{E}(Y^{a_D=1, a_Y=0} \mid D^{a_D=1}=0, L^{a_D=1}=l) =
\mathbb{E}(Y^{a_D=0, a_Y=0} \mid D^{a_D=0}=0, L^{a_D=0}=l)\\
&\overset{(\text{A\ref{ass:swig_augmented_4arm_L}})}{=}
\mathbb{E}(Y \mid D=0, L=l, A_D=0, A_Y=0)\overset{(\text{A\ref{ass:treatment_decomposition}})}{=} \mathbb{E}(Y \mid D=0, L=l, A=0).
\end{align*}  

Finally,
\begin{align*}
&CSE(a')=\mathbb{E}(Y^{a_D=a',a_Y=1} \mid D^{a_D=a'}=0 )-\mathbb{E}(Y^{a_D=a',a_Y=0} \mid D^{a_D=a'}=0 )\\
&=\sum_l \{\mathbb{E}(Y^{a_D=a',a_Y=1} \mid D^{a_D=a'}=0, L^{a_D=a'}=l) - \mathbb{E}(Y^{a_D=a',a_Y=0} \mid D^{a_D=a'}=0, L^{a_D=a'}=l)\} \\
&\quad \Pr(L^{a_D=a'}=l \mid D^{a_D=a'}=0)\\
&=\sum_l \{\mathbb{E}(Y \mid D=0, L=l, A=1) - \mathbb{E}(Y \mid D=0, L=l, A=0)\} \Pr(L^{a_D=a'}=l \mid D^{a_D=a'}=0)\\
&\overset{(\text{A\ref{ass:swig_augmented_4arm_L}})}{=}\sum_l \{\mathbb{E}(Y \mid D=0, L=l, A=1) - \mathbb{E}(Y \mid D=0, L=l, A=0)\} \\
&\quad \Pr(L=l \mid D=0, A_D=a', A_Y=a')\\
&\overset{(\text{A\ref{ass:treatment_decomposition}})}{=}\sum_l \{\mathbb{E}(Y \mid D=0, L=l, A=1) - \mathbb{E}(Y \mid D=0, L=l, A=0)\} \Pr(L=l \mid D=0, A=a').\\
\end{align*}
Assumptions \ref{ass:positivity} and \ref{ass:positivity_L} ensure that the conditional quantities are well-defined.
\end{proof}

\newpage

\section{Relevance to instrumental variable methods} \label{app:instrumental_link}

Consider again the example in section ``Independent mechanisms for adherence''. We assumed that Assumption \ref{ass:swig_augmented}, \ref{ass:positivity}, and \ref{ass:D_Asufficient_po} hold. If Assumption \ref{ass:monotonicityD_A} (monotonicity) and Assumption \ref{ass:crossworld_ind} (the cross-world independence given $U$) also hold, then the $SACE$ is equal to the association between assigned treatment and outcome among participants who adhered (Proposition \ref{prop:sace_identification}). 
In this setting, the $SACE$ represents the effect of the assigned treatment among compliers ($D^{a=0}=D^{a=1}=0$), an estimand that is often targeted by instrumental variable methods\cite{angrist1995identification}. However, our functional in Proposition \ref{prop:sace_identification} is not the same as the traditional instrumental variable functional. This discrepancy can be explained by the fact that our result relies on assumptions that are different than the conventional instrumental variable assumptions\cite{angrist1995identification}. In particular, the monotonicity assumption we invoke (Assumption \ref{ass:monotonicityD_A}) is different from the conventional monotonicity assumption invoked in instrumental variable analyses. To explain the difference between our monotonicity assumption (Assumption \ref{ass:monotonicityD_A}) and the traditional monotonicity assumption, consider the indicator variable $M$, taking value $1$ if the participant actually used NRT in the first 4 weeks, and $M=0$ if the participant actually used e-cigarette in the first 4 weeks. To ensure that $M$ has the same support as $A$, we make the (strong) assumption that individuals who do not adhere to the assigned treatment, will then take the non-allocated treatment (and that no other treatment options exist). For example, if an individual assigned to NRT does not want to use NRT anymore, they will start using e-cigarettes.

In this setting, $M^{a=a'}=(1-a') \cdot D^{a=0}+a' \cdot (1-D^{a=1})$. The exclusion restrictions assumptions implied by the causal model (Assumption \ref{ass:swig_augmented}) and the determinism (Assumption \ref{ass:D_Asufficient_po}) impose some constraints on the possible values of the counterfactual variables. In the table, we only include counterfactual response types that are compatible with Assumptions \ref{ass:swig_augmented} and \ref{ass:D_Asufficient_po}.\footnote{We thank Professor Thomas Richardson, University of Washington, Seattle, USA for suggesting that we report such a table.}

\vspace{2mm}

\begin{table}[ht]
\centering
\begin{tabular}{l|cc|cc|cc|l}
\hline
Type & $D^{a=1}_A$ & $D^{a=0}_A$ & $D^{a=1}$ & $D^{a=0}$ & $M^{a=1}$ & $M^{a=0}$ & Names in IV literature\\ 
\hline
type 1 & \multirow{2}{*}{0} & \multirow{2}{*}{0} & 0 & 0 & 1 & 0 & compliers \\
        &                     &                     & 1 & 1 & 0 & 1 & defiers \\
\hline
type 2 & \multirow{2}{*}{1} & \multirow{2}{*}{0} & 1 & 0 & 0 & 0 & never takers \\
        &                     &                     & 1 & 1 & 0 & 1 & defiers \\
\hline
\color{blue} type 3 & \color{blue}  \multirow{2}{*}{0} & \color{blue}  \multirow{2}{*}{1} &\color{blue}  0 &\color{blue}  1 &\color{blue}  1 &\color{blue}  1 &\color{blue}  always takers \\
        &                     &                     &\color{blue}  1 &\color{blue}  1 &\color{blue}  0 &\color{blue}  1 &\color{blue}  defiers \\
\hline
type 4 & 1 & 1 & 1 & 1 & 0 & 1 & defiers \\
\hline
\end{tabular}
\caption*{Table representing the counterfactual response types of the variable $D_A$, and the corresponding possible values for the variables $D^{a=1},D^{a=0},M^{a=1},M^{a=0}$ under our causal model. The counterfactual response type ruled out by Assumption \ref{ass:monotonicityD_A} is colored in blue.}
\end{table}

Assumption \ref{ass:monotonicityD_A} rules out the existence of individuals with response type 3. In our example, it means that all individuals who do not deviate from assignment due to side effects when assigned to NRT ($D^{a=1}_A=0$), do not deviate from assignment due to side effects when assigned to e-cigarette ($D^{a=0}_A=0$). 

When we also invoke Assumption \ref{ass:monotonicityD_A}, then the group referred to in the instrumental variable literature as ``always takers'' ($M^{a=1}=1, M^{a=0}=1$) cannot exist. That is, no participant would always use NRT, regardless of treatment assignment; a participant who would not adhere when assigned e-cigarettes,  would not adhere if they were assigned NRT. If the only way treatment assignment affects adherence is through side effects, and e-cigarettes have less side effects than NRT, then, individuals who fail to adhere to e-cigarettes have characteristics that make them always non-adherent. Therefore, a participant who takes NRT when assigned e-cigarettes, will not take NRT when assigned to NRT. On the other hand, Assumption \ref{ass:monotonicityD_A} does not rule out the presence of defiers, which are precluded by conventional identification assumptions for complier average effects. 

The assumption that individuals who fail to adhere to the assigned treatment always use the non-allocated alternative is unrealistic; we only did this to align with the conventional instrumental variable assumptions. For example, in the actual trial\cite{hajek2022electronic}, only $9.8\%$ of participants in the NRT arm and $1.9\%$ in the e-cigarettes arm used non-allocated products at 4 weeks, not representing all non-adherent participants. As mentioned in the main text,  our identification results do not require this assumption: non-adherence $D=1$ can indicate that the participant takes a different stop-smoking product than the one assigned, that the participant refuses any help, or that they drop out of the study. When multiple options exist for non-adherence and $M$ does not have the same two levels as $A$, identifying the complier average causal effect using instrumental variables is not trivial.

\end{appendices}
\end{document}